\documentclass{amsart}
\usepackage{amssymb,amsmath}
\usepackage[latin1]{inputenc}
\usepackage[OT1]{fontenc}
\newtheorem{thm}{Theorem}[section]
\newtheorem{prop}[thm]{Proposition}
\newtheorem{cor}[thm]{Corollary}

\newtheorem{defn}[thm]{Definition}

\newtheorem{prob}[thm]{Problem}
\newtheorem{rem}[thm]{Remark}

\def\O{\mathcal{O}}
 \def\kk{\mathbb{K}}
 \def\II{\mathbb{I}}
 
 \def\NN{\mathbb{N}}
 
 \def\Tl{\tilde{T}}
 
 \def\ML{\mathcal{L}}

 \def\<{\langle}
 \def\>{\rangle}
 
 \def\Oc{\mathcal{O}}
 \def\Osft{\tilde{\mathcal{O}}}
 \newcommand{\Unit}[1]{\mathfrak{U}_{#1}}

\begin{document}
\title{Superfast solution of Toeplitz systems based on syzygy reduction}

\author{Houssam Khalil}
\address
{{Houssam Khalil}, Institut Camille Jordan, universit\'e Claude Bernard Lyon 1, 43 boulevard du 11 novembre 1918, 69622 Villeurbanne cedex France}
\email{khalil@math.univ-lyon1.fr}

\author{Bernard Mourrain}
\address
{{Bernard Mourrain}, INRIA, GALAAD team, 2004 route des Lucioles, BP 93, 06902 Sophia Antipolis Cedex, France}
\email{mourrain@sophia.inria.fr}

\author{Michelle Schatzman}
\address
{{Michelle Schatzman},
Institut Camille Jordan, universit\'e Claude Bernard Lyon 1, 43 boulevard du 11 novembre 1918, 69622 Villeurbanne cedex France}
\email{schatz@math.univ-lyon1.fr}

\begin{abstract}
  We present a new superfast algorithm for solving Toeplitz
  systems. This algorithm is based on a relation between the solution
  of such problems and syzygies of polynomials or moving lines. We
  show an explicit connection between the generators of a Toeplitz
  matrix and the generators of the corresponding module of
  syzygies. We show that this module is generated by two elements and
  the solution of a Toeplitz system $T\,u=g$ can be reinterpreted as
  the remainder of a vector depending on $g$, by these two
  generators. We obtain these generators and this remainder with
  computational complexity $\O(n\log^2 n)$ for a Toeplitz matrix of
  size $n\times n$.

\end{abstract}

%

\maketitle

\section{Introduction}
Structured matrices appear in various domains, such as scientific
computing, signal processing, \dots They usually express, in a
linearize way, a problem which depends on fewer parameters than the
number of entries of the corresponding matrix. An important area of
research is devoted to the development of methods for the treatment of
such matrices, which depend on the parameters defining them.

Among well-known structured matrices, Toeplitz and Hankel structures
have been intensively studied \cite{MR782105,MR1355506}. Nearly
optimal algorithms are known for their multiplication by a vector and
the solution of linear systems, for such structure. Namely, if $A$ is
a Toeplitz matrix of size $n$, multiplying it by a vector or solving a
linear system with $A$ requires $\Osft(n)$ arithmetic operations
(where $\Osft(n)=\Oc(n \log^{c}(n))$ for some $c>0$)
\cite{LDRBA80,MR1871324}. Such algorithms are called superfast, in
opposition with fast algorithms requiring $\Oc(n^{2})$ arithmetic
operations.

The fundamental ingredients in these algorithms are the so-called
generators \cite{MR1355506}, encoding the minimal information stored
in these matrices, and on which the matrix transformations are
translated.  The correlation with other types of structured matrices
has also been well developed in the literature
\cite{MR1843842,MR1755557}, allowing to treat efficiently other
structures such as Vandermonde or Cauchy-like structures.

Such problems are strongly connected to polynomial problems
\cite{LDRFuh96b,MR1289412}. For instance, the product of a Toeplitz
matrix by a vector can be deduced from the product of two univariate
polynomials, and thus can be computed efficiently by
evaluation-interpolation techniques, based on FFT. The inverse of a
Hankel or Toeplitz matrix is connected to the Bezoutian of the
polynomials associated to their generators. Such a construction is
related to Gohberg-Semencul formula \cite{LDRGS72} (or Trench
algorithm \cite{LDRTre64}), which describes the inverse  of a Toeplitz matrix in terms of the 
solution of two specific Toeplitz systems (see also Gohberg-Prupnick formula \cite{LDRGK72}).

Most of these methods involve univariate polynomials. So far,
few investigations have been pursued for the treatment of multilevel
structured matrices \cite{LDRTyr85,khalil}, related to multivariate
problems. Such linear systems appear for instance in resultant or in
residue constructions, in normal form computations, or more generally
in multivariate polynomial algebra. We refer to \cite{MR1762401} for a
general description of multi-structured
matrices and their correlations with multivariate
polynomials. Surprisingly, these multivariate structure also appear
in numerical scheme and preconditionners \cite{khalil}. A main
challenge here is to devise superfast algorithms of complexity
$\Osft(n)$ for the solution of multi-structured systems of size
$n$. 

In this paper, we re-investigate the solution of Toeplitz systems $T\,
u =g$ from a new point of view which can be generalized to two-level
Toeplitz systems. We correlate the solution of such problems with
syzygies of polynomials. We show an explicit connection between the
generators of a Toeplitz matrix and the generators of the
corresponding module of syzygies. We show that this module is
generated by two elements of degree $n$ and the solution of $T\,u=g$
can be reinterpreted as the remainder of an explicit polynomial vector
depending on $g$, by these two generators. We give two algorithms,
with computational complexity $\O(n\log^2n)$, to compute the
generators of the module of syzygies. We give finally an algorithm,
with computational complexity $\O(n\log^2n)$, for the division of the
generators by the polynomial vector depending on $g$.  Our new syzygy
approach can be connected with Pad\'e approximation method developed
in \cite{LDRBGY80} to compute efficiently particular solutions of
Toeplitz linear system. But we replace the computation of generators
of structured matrices by the computation of generators of a syzygy
module and the solution of the linear system from particular solutions
by Euclidean reduction by the generators of the syzygy module.

Let $R=\kk[x]$. For $n \in \NN$, we denote by $\kk[x]_{n}$ the vector
space of polynomials of degree $\le n$. Let $L=\kk[x,x^{-1}]$ be the
set of Laurent polynomials in the variable $x$. For any polynomial
$p=\sum_{i=-m}^{n} p_{i}\, x^{i} \in L$, we denote by $p^{+}$ the sum
of terms with non-negative exponents: $p^{+}=\sum_{i=0}^{n} p_{i}\,
x^{i}$ and by $p^{-}$, the sum of terms with strictly negative
exponents: $p^{-}=\sum_{i=-m}^{-1} p_{i}\, x^{i}$. We have $p=p^{+}
+p^{-}$.

For $n\in \NN$, we denote by $\Unit{n}=\{\omega; \omega^{n}=1\}$ the set of roots of unity of order $n$.

For a vector $u=(u_0,\dots,u_{k-1})^T\in\kk^k$, we denote by $u(x)$
the polynomial of degree $k-1$ given by
$u(x)=\sum_{i=0}^{k-1}u_ix^i$. Conversely, if
$v(x)=\sum_{i=0}^{k-1}v_ix^i$ is a polynomial of degree $k-1$, we
denote by $v$ the vector of length $k$ of coefficients of $v(x)$.

If no confusion arises, we may also use $v$ to denote the polynomial $v(x)$.

\section{Sygygies and Toeplitz matrices}
Let $T\in\kk^{n\times n}$ be an $n\times n$ Toeplitz matrix. Then $T$
is of the following form: 
\begin{equation}
 \begin{pmatrix}
t_0&t_{-1}&\dots&t_{-n+1}\\
t_{1}&t_0&\ddots&\vdots\\
\vdots&\ddots&\ddots&t_{-1}\\
t_{n-1}&\dots&t_{1}&t_0
\end{pmatrix}.
\end{equation}
Let $g=(g_0,\dots,g_{n-1}) \in \kk^{n}$ be a vector of length $n$. We
are interested in the following problem:
\begin{prob}\label{pb:initial}
 Find $u=(u_0,\dots,u_{n-1}) \in
\kk^{n}$ such that 
\begin{equation}\label{pb:toep}
T\,u=g.
\end{equation}
\end{prob}

\begin{defn}
Let $E=\{1,\dots,x^{n-1}\}$, and $\Pi_{E}$ be the projection of $L$ on the vector space generated by $E$, along $\<x^{n},x^{n+1},\ldots\>$.   
\end{defn}
\begin{defn}
From the matrix $T$ and the vectors $g$ and $u$ we define the following polynomials: 
\begin{itemize}
 \item $T(x)=\displaystyle\sum _{i=-n+1}^{n-1}t_ix^i,$
 \item $\tilde{T}(x)=\displaystyle\sum_{i=0}^{2n-1}\tilde{t}_ix^i$ with
$\tilde{t}_i=\left\{
\begin{array}{ll} t_i&\textrm{ if } i< n\\ 
t_{i-2n}&\textrm{ if } i\ge n
\end{array}\right.$,
 \item  $u(x)=\displaystyle\sum_{i=0}^{n-1}u_ix^i,\:g(x)=\sum_{i=0}^{n-1}g_ix^i$.
\end{itemize}
\end{defn}

\noindent{}Notice that $T(x)$ is a Laurent polynomial and that
$\tilde{T}(x)$ is a polynomial of degree $2n-1$. By construction, we have the following properties:
\begin{prop}
$\tilde{T}= T^{+} + x^{2\,n}\, T^{-}$ and $T(w)=\tilde{T}(w)$ if $w \in \Unit{2\,n}$.
\end{prop}
\begin{proof}
We can deduce directly, from the definition of $T(x)$ and $\Tl(x)$, that $\tilde{T}= T^{+} + x^{2n}\, T^{-}$. Moreover, since $w^{2n}=1$ and  $\Tl(x)=T^{+}(x)+x^{2n}T^{-}(x)$, then $\Tl(w)=T^{+}(w)+T^{-}(w)=T(w)$.
\end{proof}

According to Proposition $2.1.2$ of \cite{MR1762401}, we have the following relation between the problem \ref{pb:initial} and polynomials:
\begin{prop}\label{Toepolyn}
We have
$$
T\,u=g\Leftrightarrow\Pi_{E}(T(x)u(x))=g(x).
$$
\end{prop}

As $\Pi_{E}(T(x)u(x))$ is the polynomial $T(x)u(x)$ from which we remove terms of negative degree and of degree $\geq n$, then we can write $T(x)\, u(x)$ as following:
\begin{prop}\label{transf}
\begin{equation}
T(x)\, u(x) = \Pi_{E}(T(x)u(x))+ x^{-n} A(x) + x^{n} B(x),
\end{equation}
where $A(x)\in\kk[x]_{n-1}$ and $B(x)\in\kk[x]_{n-2}$.
\end{prop}
\begin{proof}
By expanding $T(x)u(x)$ we can write
\begin{eqnarray*}
 T(x)u(x)&=&\Pi_{E}(T(x)u(x))+(\alpha_{-n+1}x^{-n+1}+\dots+\alpha_{-1}x^{-1}) + (\alpha_{n}x^n+\dots+\alpha_{2n-2}x^{2n-2})\\
&=&\Pi_{E}(T(x)u(x))+x^{-n}(\alpha_{-n+1}x+\dots+\alpha_{-1}x^{n-1}) + x^{n}(\alpha_{n}+\dots+\alpha_{2n-2}x^{n-2})\\
&=&\Pi_{E}(T(x)u(x))+ x^{-n} A(x) + x^{n} B(x)
\end{eqnarray*}
\end{proof}

Therefore, according to Proposition \ref{Toepolyn} and  Proposition \ref{transf}, if $u$ is solution of $Tu=g$ then there exist two polynomials $A(x)$ and $B(x)$ in $\kk[x]_{n-1}$ such that
\begin{equation}\label{Tsyz}
T(x)u(x)- x^{-n} A(x) - x^{n} B(x) = g(x).
\end{equation}
By evaluation at the roots $\omega \in \Unit{2n}$, and since $\omega^{-n}= \omega^{n}$ and $\Tl(\omega)=T(\omega)$ for $\omega\in \Unit{2n}$, we have
$$ 
\Tl(\omega) u(\omega) + \omega^{n} v(\omega) = g(\omega),
\forall \omega \in \Unit{2n}(\omega),
$$
where $v(x)= -A(x)-B(x)$ of degree $\le n-1$. Therefore the polynomial $\Tl(x) u(x) + x^{n} v(x) - g(x)$ is multiple of $x^{2n}-1$. We deduce that there exists $w(x)\in\kk[x]$ such that
\begin{equation}\label{TTsyz}
\Tl(x) u(x) + x^{n} v(x) + (x^{2n}-1) w(x)= g(x).
\end{equation}
Notice that $w(x)$ is of degree $\le n-1$ because $(x^{2n}-1)\, w(x)$ is of degree $\le 3n-1$.


\subsection{Syzygies}

The solutions of Equation \eqref{TTsyz} is a particular case of the
following problem, related to interesting questions in Effective Algebraic Geometry.
\begin{prob}\label{pb:mvlines}
Given three polynomials $a, b, c \in R$  respectively of degree $<l, <m, <n$,
find three polynomials  $p, q, r \in R$ 
of degree $< \nu-l, <\nu-m, <\nu-n$, such that
\begin{equation} \label{eq:mvlines}
a(x)\, p(x) + b(x)\, q(x) + c(x)\, r(x) =0.
\end{equation}
\end{prob}
The polynomial vector $(p,q,r) \in\kk[x]^{3}$ is called a {\em syzygy} of $(a,b,c)$.
We denote by $\ML(a,b,c)$ the set of syzygies $(p,q,r)\in\kk[x]^{3}$ of $(a,b,c)$,
i.e. the solutions of \eqref{eq:mvlines}. It is a $\kk[x]$-module of
$\kk[x]^{3}$ and it is called the module of syzygies of $(a,b,c)$. The
solutions of Problem \ref{pb:mvlines} are $\ML(a,b,c)
\cap\kk[x]_{\nu-l-1}\times\kk[x]_{\nu-m-1}\times\kk[x]_{\nu-n-1}$.

Given a new polynomial $d(x)\in \kk[x]$, we denote by $\ML(a,b,c;d)$ the set of $(p,q,r)\in\kk[x]^{3}$ such that 
\begin{equation}
a(x)\, p(x) + b(x)\, q(x) + c(x)\, r(x) = d(x).
\end{equation}

\begin{thm}
For any non-zero vector of polynomials $(a,b,c)\in \kk[x]^{3}$, the $\kk[x]$-module $\ML(a,b,c)$ is free of rank $2$.
\end{thm}
\begin{proof}
By the Hilbert's theorem, the ideal $I$ generated by $(a,b,c)$ has a
free resolution of length at most $1$ (see \cite[chap. 6]{MR1639811}), that is of the form:
$$
0\rightarrow\kk[x]^p\rightarrow \kk[x]^3\rightarrow \kk[x] \rightarrow \kk[x]/I \rightarrow 0.
$$ 
As $I\neq 0$, for dimensional reasons, we must have $p-3+1=0$, then $p=2$.
\end{proof}

\begin{defn}
For a polynomial vector $p=(p_1,\dots,p_k)\in\,\kk[x]^k$, we
define $$\deg(p_{1}, \ldots, p_{k})=\max(\deg(p_1),\dots,\deg(p_k)).$$
\end{defn}

\begin{defn} Assume that $\deg (p,q,r)\leq \deg(p',q',r')$. A $\mu$-base of $\ML(a,b,c)$ is a basis $\{(p,q,r),\,(p',q',r')\}$ of $\ML(a,b,c)$, with $\deg(p,q,r)=\mu$.
\end{defn}

We have the following relation between the degrees of the two elements of a basis of $\ML(a,b,c)$:
\begin{prop}\label{degree}
Let $\{(p_{1},q_{1},r_{1}),\,(p_{2},q_{2},r_{2})\}$ be a basis of $\ML(a,b,c)$, $\mu_1=\deg(p_1,q_1,r_1)$ and $\mu_2=\deg(p_2,q_2,r_2)$. We have $\deg(a,b,c)=\mu_1+\mu_2$.
\end{prop}
\begin{proof}
We have
\begin{equation*}
0 \rightarrow \kk[x]_{\nu-d-\mu_1} \oplus \kk[x]_{\nu-d-\mu_2} \rightarrow \kk[x]_{\nu-d}^3\rightarrow \kk[x]_{\nu} \rightarrow \kk[x]_{\nu}/(a,b,c)_{\nu} \rightarrow 0,
\end{equation*}
for $\nu \gg 0$. As the alternate sum of the dimension of the $\kk$-vector spaces is zero and $\kk[x]_{\nu}/(a,b,c)_{\nu}$ is $0$ for $\nu \gg 0$, we have
$$ 
0 = 3\,(d-\nu-1)  +\nu -\mu_1- d +1 + \nu -\mu_2 -d +1 + \nu +1 =  d -\mu_1 -\mu_2.
$$
\end{proof}

\subsection{The module $\ML(\Tl(x), x^{n}, x^{2n}-1)$}
Returning to the initial problem, we saw that if $u$ is solution of $Tu=g$ then there exist two polynomials $v(x)$ and $w(x)$ in $\kk[x]_{n-1}$ such that $(u(x),v(x),w(x))$ $\in\ML(\Tl(x), x^{n}, x^{2n}-1;g(x))$.

By the proposition \ref{degree}, if $(p,q,r)$ and$(p',q',r')$ form a
basis of $\ML(\Tl(x), x^{n}, x^{2n}-1)$ of degree $\mu_1$ and $\mu_2$
respectively then we have $\mu_1+\mu_2=2\,n$.
We are going to show now that in fact $\ML(\Tl(x), x^{n}, x^{2n}-1)$
has a $n$-basis, that is a basis of two elements of degree $\mu_1=\mu_2=n$: 
\begin{prop}\label{prop:ML}
The $\kk[x]$-module $\ML(\Tl(x), x^{n}, x^{2n}-1)$ has a $n$-basis.
\end{prop}
\begin{proof}
Consider the linear map
\vspace{-0.5cm}
\begin{eqnarray}\label{syzfunction1}
\kk[x]_{n-1}^3 &\rightarrow & \kk[x]_{3n-1}\\
(p(x),q(x),r(x))  & \mapsto & \Tl(x) p(x) + x^{n} q(x) +(x^{2n}-1) r(x),\nonumber
\end{eqnarray}
which $3n \times 3n$ matrix is of the form 
\begin{equation}\label{form:S}
S:= \left(
\begin{array}{c|c|c}
T_{0} &\mathbf{0}  &  -\II_{n} \\ 
T_{1} & \II_{n}    &  \mathbf{0} \\
T_{2} & \mathbf{0} &  \ \, \II_{n}  \\ 
\end{array}
\right),
\end{equation}
where $T_{0}, T_{1}, T_{2}$ are the coefficient matrices of $(\Tl(x)$, $x\, \Tl(x)$, $\ldots,$ $x^{n}\Tl(x))$, respectively for the
list of monomials $(1,\ldots,x^{n-1})$, $(x^{n},\ldots,x^{2n-1})$, $(x^{2n},\ldots, x^{3n-1})$. Notice in particular that $T= T_{0}+T_{2}$.

Reducing the first block $(T_{0}| \mathbf{0} | -\II_{n})$ by the last block $(T_{2}| \mathbf{0} | \II_{n})$, we replace it by the block
$(T_{0}+T_{2}| \mathbf{0} | \mathbf{0})$, without changing the rank of $S$. As $T=T_{0}+T_{2}$ is invertible, this shows that the matrix
$S$ is of rank $3n$. Therefore $\ker (S)=0$ and there is no syzygies in degree $n-1$.

As the sum $2n=\mu_1+\mu_2$, where $\mu_1,\mu_2$ are the degrees of a pair of  generators of $\ML(\Tl(x), x^{n}, x^{2n}-1)$, and as $\mu_1\geq n$ and $\mu_2\geq n$, we have $\mu_1=\mu_2=n$. Moreover, $\ML(\Tl(x), x^{n}, x^{2n}-1)$ is free of rank $2$. Thus there exist two linearly independent syzygies $(u_1,v_1,w_1)$, $(u_2,v_2,w_2)$ of degree $n$, which generate $\ML(\Tl(x), x^{n}, x^{2n}-1)$.
\end{proof}

A similar result can also be found in \cite{MR1871324}, but the proof much longer than this one, is based on interpolation techniques and explicit computations.

Let us now describe how to construct explicitly two generators 
$(u_1,v_1,w_1)$, $(u_2,v_2,w_2)$ of $\ML(\Tl(x), x^{n}, x^{2n}-1)$ of
degree $n$. As $\Tl(x)$ is of degree $\le 2\,n -1$ and the map \eqref{syzfunction1} is surjective, there exists $(u,v,w) \in \kk[x]_{n-1}^3$ such
that \begin{equation}\label{base1} \Tl(x) u(x) + x^n v(x) +
  (x^{2\,n}-1)\, w = \Tl(x) x^n.
\end{equation}
We deduce that $(u_1,v_1,w_1)=(x^n-u, -v, -w) \in \ML(\Tl(x), x^{n}, x^{2n}-1)$.

Since there exists $(u',v',w') \in \kk[x]_{n-1}^3$ such that 
\begin{equation}\label{base2}
\Tl(x) u'(x) + x^n v'(x) + (x^{2\,n}-1)\, w' =1 = x^n\, x^n - (x^{2\,n}-1),
\end{equation}
we deduce that $(u_2,v_2,w_2)=(-u',x^n -v', -w' - 1) \in \ML(\Tl(x), x^{n}, x^{2n}-1)$.

Now, T and are linearly independent since by construction, 
The coefficient vectors of $x^{n}$ in $(u_1,v_1,w_1)$ and
$(u_2,v_2,w_2)$ are respectively $(1,0,0)$ and $(0,1,0)$, which shows
that vectors $(u_1,v_1,w_1)$, $(u_2,v_2,w_2)\in \ML(
\Tl(x),x^{n},x^{2n}-1)\cap \kk[x]_{n}$ are linearly independent.
Therefore, they form a basis of $\ML(\Tl(x),x^{n},x^{2n}-1)$.

Now we can prove our aim theorem:
\begin{thm}\label{division}
The vector $u$ is solution of \eqref{pb:toep}  if and only if  there exist $v(x)$ and $w(x)$ in $\kk[x]_{n-1}$ such that 
$$
(u(x), v(x), w(x)) \in \ML(\tilde{T}(x), x^{n}, x^{2n}-1; g(x) )
$$
\end{thm}
\begin{proof}
If $u$ is solution of \eqref{pb:toep}, we see that there exist $v(x)\in\kk[x]_{n-1}$ and $w(x)\in\kk[x]_{n-1}$ such that
$$ 
\Tl(x) u(x) + x^{n} v(x) + (x^{2n}-1) w(x)= g(x).
$$

Conversely,  a solution $(u(x), v(x), w(x)) \in
\ML(\tilde{T}(x),x^{n},x^{2n}-1; g(x) )\cap \kk[x]_{n-1}^{3}$ implies
that $(u,v,w)\in \kk^{3\,n}$ is a solution of the 
linear system:
$$ 
S 
\, \left(
\begin{array}{c}
u\\
v\\
w\\
\end{array}
\right)
=
\left(
\begin{array}{c}
g\\
0\\
0\\
\end{array}
\right),
$$
where $S$ is has the block structure \eqref{form:S}, so that $T_{2}\, u + w =0$ and $T_{0}\, u - w = (T_{0}+T_{2}) u=g$. As we have $T_{0}+T_{2}=T$, the vector $u$ is a solution of \eqref{pb:toep}, which ends the proof of the theorem.
\end{proof}

Computing the inverse of a Toeplitz matrix $T$ is equivalent to computing
the first and the last column of $T^{-1}$, based on 
Gohberg-Semencul decomposition  (see \cite{LDRHR84,MR1179345,MR1491603,LDRKC89} for more details about
  Gohberg-Semencul decomposition). 

We are going to show that the solutions of Equations
\eqref{base1} and \eqref{base2} which gives us
the $n$-basis $\{(u_1,v_1,w_1),(u_2,v_2,w_2)\}$ is related to the
solution of two specific Toeplitz linear systems.

\begin{prop}\label{prop:2.15}
Let $(u(x),v(x),w(x))$ and $(u'(x),v'(x),w'(x))$ be in $\kk_{n-1}[x]^3$ such that
$$
\left\{\begin{array}{l}
\Tl(x) u(x) + x^n v(x) + (x^{2\,n}-1)\, w(x) = \Tl(x) x^n,\\
\Tl(x) u'(x) + x^n v'(x) + (x^{2\,n}-1)\, w'(x) =1.
\end{array}\right.
$$
Then $Tu'=e_1$ and $Tu=ZTe_n$, with $Z$ is the lower shift matrix of size $n$.
\end{prop}
\begin{proof}
As $u'(x),\,v'(x),\, w'(x)$ and $1$ are of degree $\leq n-1$, then, by
Theorem \ref{division}, $\Tl(x) u'(x) + x^n v'(x) + (x^{2\,n}-1)\,
w'(x) =1$ is equivalent to $Tu'=e_1$ ($e_1(x)=1$) and $u'$ is the first column of $T^{-1}$.

We have $\Tl(x)=T_+(x)+x^{2n}T_-(x)$, then 
$$\Tl(x) u(x) + x^n v(x) + (x^{2n}-1)w(x) =x^nT_+(x)+x^n((x^{2n}-1)T_-(x)+T_-(x)).$$
Therefore,
$$\Tl(x)u(x)+x^n(v(x)-T_+(x))+(x^{2n}-1)(w(x)-x^nT_-(x))=x^nT_-(x).$$
As $x^nT_-(x)$ is of degree $\leq n-1$ and is the polynomial associated with the vector $ZTe_n$, by Theorem \ref{division}, $u$ is such that $Tu=ZTe_n$.
\end{proof}

Notice that $u$ is not the last column of $T^{-1}$, but we can use $u$ and $u'$ to compute it (see \cite{LDRHR84}). 
Therefore, defining a $n$-basis of $\ML(\Tl(x),x^{n},x^{2n}-1)$ 
from  the solution of Equations \eqref{base1} and \eqref{base2} is
equivalent to computing the Gohberg-Semencul decomposition of
$T^{-1}$. 

In the following section, we reduce translation of the
solution of $Tu=g$ to an Euclidean division, based on our
decomposition, instead of multiplying $g$ by triangular
Toeplitz matrices, based on Gohberg-Semencul decomposition. The
advantage of our decomposition is that we can generalized it to
two-level problems, which allows us to describe a ``Gohberg-Semencul''
decomposition of Toeplitz-block-Toeplitz matrices.
 
\section{Euclidean division}
In this section, we show how to obtain the solution vector
$(u(x),v(x),w(x))\in \ML(\Tl(x), x^{n}, x^{2\,n}-1;g(x))\cap \kk[x]_{n}^{3}$ from a $n$-basis of $\ML(\Tl(x),x^{n},x^{2n}-1)$
and a particular solution in $\ML(\Tl(x), x^{n}, x^{2\,n}-1;g(x))$.

From Theorem \ref{division} we deduce the two following corollaries:
\begin{cor}
For all $g(x)\in \kk_{n-1}[x]$, the set $\ML(\tilde{T}(x), x^{n}, x^{2n}-1;g(x))\cap\kk_{n-1}^3[x]$ has exactly one element.
\end{cor}
\begin{proof}
As $T$ is invertible, there exists a unique $u$ such that $Tu=g$. From the theorem \ref{division}, there exists $v(x),\,w(x)$ of degree $\leq n-1$, such that $(u(x),v(x),w(x))\in\ML(\tilde{T}(x), x^{n}, x^{2n}-1;g(x))\cap\kk_{n-1}^3[x]$. 

The uniqueness is also obvious: if
$(u'(x),v'(x),w'(x))\in\ML(\tilde{T}(x), x^{n},
x^{2n}-1;g(x))\cap\kk_{n-1}^3[x]$, then
$(u(x),v(x),w(x))-(u'(x),v'(x),w'(x))\in\ML(\tilde{T}(x), x^{n},
x^{2n}-1)\cap\kk_{n-1}^3[x]$ which equal $\{(0,0,0)\}$ (see the
demonstration of the proposition \ref{prop:ML}). Then
$(u(x)$, $v(x)$, $w(x))$ $=(u'(x),v'(x),w'(x))$.
\end{proof}

\begin{cor}
Let $\{(u_1,v_1,w_1),(u_2,v_2,w_2)\}$ be a $n$-basis of
$\ML(\Tl(x),x^{n},x^{2n}-1)$. Let $(p,q,r)$ be in
$\ML(\Tl(x),x^{n},x^{2n}-1; g(x))$. There exists a unique
$(u,v,w)\in\ML(\tilde{T}(x), x^{n}, x^{2n}-1;g(x)) \cap\kk_{n-1}^3[x]$
and a unique pair of polynomials $p_1$ and $p_2$ such that
$$
\begin{pmatrix}p\\q\\r\end{pmatrix}=
p_1\begin{pmatrix}u_1\\v_1\\w_1\end{pmatrix}+
p_2\begin{pmatrix}u_2\\v_2\\w_2\end{pmatrix}+
\begin{pmatrix}u\\v\\w\end{pmatrix}.
$$
This decomposition is called the division of $(p,q,r)$ by $(u_1,v_1,w_1)$ and $(u_2,v_2,w_2)$.
\end{cor}
\begin{proof}
From the previous corollary, there exist a unique element in
$\ML(\tilde{T}(x), x^{n}, x^{2n}-1;g(x)) \cap\kk_{n-1}^3[x]$, let
$(u,v,w)$ be this element. As $\{(u_1,v_1,w_1),(u_2,v_2,w_2)\}$ is a
$n$-basis of $\ML(\Tl(x),x^{n},x^{2n}-1)$, and as
$(p,q,r)-(u,v,w)\in\ML(\tilde{T}(x), x^{n}, x^{2n}-1)$, then there
exist a unique pair of polynomials unique $p_1$ and $p_2$ such that
$$
\begin{pmatrix}p\\q\\r\end{pmatrix}-
\begin{pmatrix}u\\v\\w\end{pmatrix}=
p_1\begin{pmatrix}u_1\\v_1\\w_1\end{pmatrix}+
p_2\begin{pmatrix}u_2\\v_2\\w_2\end{pmatrix}
$$
\end{proof}

As a consequence of the two corollaries, we have the following important property:

\begin{thm}\label{divi}
Let $\{(u_1,v_1,w_1),(u_2,v_2,w_2)\}$ be a $n$-basis of
$\ML(\Tl(x),x^{n},x^{2n}-1)$, and let $g\in\kk^n$. The remainder of
the division of  $\begin{pmatrix}0\\x^n\,g\\g\end{pmatrix}$ by
$\begin{pmatrix}u_1&u_2\\v_1&v_2\\w_1&w_2\end{pmatrix}$ is the unique
element $(u,v,w)\in \ML(\tilde{T}(x), x^{n}, x^{2n}-1;g(x)) \cap\kk_{n-1}^3[x]$, and therefore $u$ is the solution of $Tu=g$. 
\end{thm}
\begin{proof}
The vector $\begin{pmatrix}0\\x^n\, g\\ -g\end{pmatrix}\in\ML(\Tl(x),
x^{n}, x^{2\,n}-1;g)$ is a particular solution. We reduce it by
$\begin{pmatrix}u_1&u_2\\v_1&v_2\\w_1&w_2\end{pmatrix}$ and obtain
$$\begin{pmatrix}u\\v\\w\end{pmatrix}=\begin{pmatrix}0\\x^n\,g\\g\end{pmatrix}-\begin{pmatrix}u_1&u_2\\v_1&v_2\\w_1&w_2\end{pmatrix}\begin{pmatrix}p\\q\end{pmatrix}, $$
where $(u,v,w)\in\kk[x]^3_{n-1}\cap\ML(\Tl(x), x^{n},
x^{2\,n}-1;g)$ is the remainder of division.  Thus $(u,v,w)$ is the unique vector $\in\kk[x]^3_{n-1}\cap\ML(\Tl(x), x^{n}, x^{2\,n}-1;g)$.
\end{proof}

A way to perform the division is to choose a $n$-basis $\{(u_1,v_1,w_1),$ $(u_2,v_2,w_2)\}$ of $\ML(\Tl(x),x^{n},x^{2n}-1)$ so that the $2\times2$ coefficient matrix of $x^n$ in
$$\begin{pmatrix}u_1(x)&u_2(x)\\v_1(x)&v_2(x) \end{pmatrix}$$ 
is invertible. In this case we can reduce the polynomial $(0,x^ng(x))$ to reach to a degree $< n-1$ and we can write in a unique way
$$
\begin{pmatrix}0\\x^ng(x)\end{pmatrix}=
p_1\begin{pmatrix}u_1\\v_1\end{pmatrix}+
p_2\begin{pmatrix}u_2\\v_2\end{pmatrix}+
\begin{pmatrix}u\\v\end{pmatrix}.
$$
By the uniqueness of the remainder in the Euclidean division, we obtain the following proposition:

\begin{prop}\label{simplificationdiv}
The first coordinate of the remainder in the division of $\begin{pmatrix}0\\x^ng\end{pmatrix}$ by $\begin{pmatrix}u&u_2\\v_1&v_2\end{pmatrix}$ is the polynomial $u(x)$ such that its associated vector $u$ is the solution of $T\,u=g$.
\end{prop}

So we set the following problem:
\begin{prob}\label{pb:division}
Given a matrix and a vector of polynomials
$\begin{pmatrix}e(x)&e'(x)\\f(x)&f'(x)\end{pmatrix}$ of degree $n$  such that
$\begin{pmatrix}e_n&e_n'\\f_n&f_n' \end{pmatrix}$ is invertible and
$\begin{pmatrix}p(x)\\q(x)\end{pmatrix}$ of degree $m\geq n$, find the remainder
of the division of $\begin{pmatrix}p(x)\\q(x)\end{pmatrix} $ by
$\begin{pmatrix}e(x)&e'(x)\\f(x)&f'(x)\end{pmatrix}$.
\end{prob}

We describe here a generalized Euclidean division algorithm to solve problem \ref{pb:division}.

Let $E(x)=\begin{pmatrix}p(x)\\q(x)\end{pmatrix}$ of degree $m$, $B(x)=\begin{pmatrix}e(x)&e'(x)\\f(x)&f'(x)\end{pmatrix}$ of degree $n\leq m$.
$E(x)=B(x)Q(x)+R(x)$ with $\deg(R(x))<n,$ and $ \deg(Q(x))\leq m-n$. Let $z=\frac{1}{x}$. We have
 \begin{eqnarray}\label{div}
&E(x)&=B(x)Q(x)+R(x)\nonumber\\
\Leftrightarrow& E(\displaystyle \frac{1}{z})&=B(\frac{1}{z})Q(\frac{1}{z})+R(\frac{1}{z})\nonumber\\
\Leftrightarrow& z^{m}E(\displaystyle \frac{1}{z})&=z^nB(\frac{1}{z})z^{m-n}Q(\frac{1}{z})+z^{m-n+1}z^{n-1}R(\frac{1}{z})\nonumber\\
\Leftrightarrow& \hat{E}(z)&= \hat{B}(z) \hat{Q}(z)+z^{m-n+1} \hat{R}(z)
\end{eqnarray}
with $ \hat{E}(z), \hat{B}(z), \hat{Q}(z), \hat{R}(z)$ are the polynomials obtained by reversing the order of coefficients of polynomials $E(z),B(z),Q(z),R(z)$.
\begin{eqnarray*}
\eqref{div}&\Rightarrow& { \hat{B}(z)}^{-1}{ \hat{E}(z)}=
\hat{Q}(z)+z^{m+n-1} { \hat{B}(z)}^{-1} {  \hat{R}(z)}\\ 
&\Rightarrow& \hat{Q}(z)={ \hat{B}(z)}^{-1} { \hat{E}(z)} \mod z^{m-n+1}
\end{eqnarray*}
The formal power series ${ \hat{B}(z)}^{-1}$ exists
because the constant coefficient of  $\hat{B}(z)$ is invertible. Thus
$\hat{Q}(z)$ is obtained by computing the first $m-n+1$ coefficients
of  $\displaystyle{ \hat{B}(z)}^{-1}{ \hat{E}(z)}$, which is obtained by computing $W(x)=\displaystyle{ \hat{B}(z)}^{-1}$, then by multiplying $W(x)$ by $ \hat{E}(z)$.

To find $W(x)=\displaystyle { \hat{B}(z)}^{-1}$ we use Newton's
iteration. Let $f(W)=\hat{B}-W^{-1}$. We have 
$$f'(W_l).(W_{l+1}-W_l)=-W_l^{-1}(W_{l+1}-W_l)W_l^{-1}=f(W_l)=\hat{B}-W_l^{-1}.$$ 
Thus we set 
$$W_{l+1}=2W_l-W_l\hat{B}W_l,$$
and $W_0=\hat{B}_0^{-1}$, which exists. Moreover, we have
\begin{eqnarray*}
W-W_{l+1}&=&W-2W_l+W_l\hat{B}W_l\\
&=&W(\mathbb{I}_2-\hat{B}W_l)^2\\
&=&(W-W_l)\hat{B}(W-W_l)
\end{eqnarray*}
Thus $W_l(x)=W(x) \mod x^{2l}$ for $l=0,\dots,\lceil\log(m-n+1)
\rceil$. 
\begin{prop}
We need $\mathcal{O}(n\log(n)\log(m-n)+m\log m)$ operations to solve problem \ref{pb:division}.
\end{prop}
\begin{proof}
We must do $\lceil\log(m-n+1) \rceil$ Newton's iteration to obtain the first $m-n+1$ coefficients of $\displaystyle { \hat{B}(z)}^{-1} =W(x)$. And each iteration requires $\mathcal{O}(n\log n)$ operations (multiplication and summation of polynomials of degree $n$). Finally, multiplication $\displaystyle { \hat{B}(z)}^{-1} \hat{E} (z)$ requires $\mathcal{O}(m\log m)$ operations.
\end{proof}

Notice that, for our problem $m=n$ and this algorithm requires
$\O(n\log^2 n)$ arithmetic operations.
In the following section, we show how to compute a $n$-basis in
$\O(n\log^2 n)$ arithmetic operations.

\section{Construction of the generators}
The canonical basis of $\kk[x]^3$ is denoted by
$\sigma_1,\sigma_2,\sigma_3$. Let $\rho_1,\,\rho_2$ be the generators of $\ML(\Tl(x),x^n,x^{2n}-1)$ of degree $n$ given by 
\begin{equation}\label{base3}
\begin{array}{l}\rho_1=x^n\sigma_1-(u,v,w)=(u_1,v_1,w_1)\\ \rho_2=x^n\sigma_2-(u',v',w')=(u_2,v_2,w_2),\end{array}
\end{equation}
where $(u,v,w),\,(u',v',w')$ are the vectors given in  \eqref{base1} and \eqref{base2}.

We describe two methods for computing $(u_1,v_1,w_1)$ and
$(u_2,v_2,w_2)$. The first one uses the Euclidean gcd algorithm, the
second one is based on the method in \cite{MR1871324}.

We recall firstly the algebraic and computational properties of the well known extended euclidean algorithm (see \cite{MR2001757}):
Given $p(x), p'(x)$ two polynomials in degree $m$ and $m'$ respectively, let
$$\begin{array}{ll}
r_0=p,\qquad&r_1=p',\qquad\\s_0=1,&s_1=0,\\t_0=0,&t_1=1.
\end{array}$$
and define
\vspace{-0.5cm}
\begin{eqnarray*}
r_{i+1}&=&r_{i-1}-q_ir_i,\\
s_{i+1}&=&s_{i-1}-q_is_i,\\
t_{i+1}&=&t_{i-1}-q_it_i,
\end{eqnarray*}
where $q_i$ results when the division algorithm is applied to $r_{i-1}$ and $r_i$, i.e. $r_{i-1}=q_ir_i+r_{i+1}$ . 

\begin{prop}
Let $l\in\NN$ such that $r_l=0$. Then $r_{l-1}=\gcd(p(x),p'(x))$.
\end{prop}

And more generally we have:
\begin{prop}\label{eea}
For all $i=1,\ldots,l$ we have
$$s_ip+t_ip'=r_i\quad \textrm{ and }\quad(s_i,t_i)=1,$$
and
$$\left\{\begin{array}{l}\vspace{2mm}
\deg r_{i+1}<\deg r_i, \quad i=1,\ldots,l-1\\ \vspace{2mm}
\deg s_{i+1}>\deg s_i\quad\textrm{ and }\quad \deg t_{i+1}>\deg t_i,\\\vspace{2mm}
\deg s_{i+1}=\deg(q_i.s_i)=\deg v-\deg r_i,\\\vspace{2mm}
\deg t_{i+1}=\deg(q_i.t_i)=\deg u-\deg r_i.
\end{array}\right.$$
\end{prop}

We can now present our algorithm. It can be found in the proof of the following theorem:
\begin{thm}
By applying the Euclidean gcd algorithm to $p(x)=x^{n-1}T$ and
$p'(x)=x^{2n-1}$ stopping in degree $n-1$ and $n-2$, we obtain $\rho_1$ and $\rho_2$ respectively.
\end{thm}
\begin{proof}
We see that $Tu=g$ if and only if there exist $a(x)$ and $b(x)$ in $\kk[x]_{n-1}$ such that 
$$\bar{T}(x)u(x)+x^{2n-1}b(x)=x^{n-1}g(x)+a(x),$$ 
where $\bar{T}(x)=x^{n-1}T(x)$ is a polynomial of degree $\leq2n-2$. In \eqref{base1} and \eqref{base2} we saw that for $g(x)=1$ $(g=e_1)$ and $g(x)=x^nT(x)$ $(g=(0,t_{-n+1},\ldots,t_{-1})^T)$ we obtain a base of $\ML(\Tl(x),x^n,x^{2n}-1)$. 

Notice that $Tu_1=e_1$ if and only if there exist $a_1(x)\in
\kk[x]_{n-2}$, $b_1(x)\in \kk[x]_{n-1}$ such that
\begin{equation}\label{eea1}\bar{T}(x)u_1(x)+x^{2n-1}b_1(x)=x^{n-1}+a_1(x),\end{equation} 
and $Tu_2=(0,t_{-n+1},\ldots,t_{-1})^T$ if and only if there exist $a_2(x) \in
\kk[x]_{n-2}$, $b_2(x) \in
\kk[x]_{n-1}$ such that
\begin{equation}\label{eea2}\bar{T}(x)(u_2(x)+x^{n})+x^{2n-1}b_2(x)=a_2(x).\end{equation} 
As $\deg a_1(x)\leq n-2$ and $\deg a_2(x)\leq n-2$, by applying the extended Euclidean algorithm in $p(x)=x^{n-1}T$ and $p'(x)=x^{2n-1}$ until we have $\deg r_l(x)=n-1$ and $\deg r_{l+1}(x)=n-2$ we obtain 
$$u_1(x)=\frac{1}{c_1}s_l(x),\quad b_1(x)=\frac{1}{c_1}t_l(x),\quad x^{n-1}+a_1(x)=\frac{1}{c_1}r_l(x),$$
and 
$$x^n+u_2(x)=\frac{1}{c_2}s_{l+1}(x),\quad b_2(x)=\frac{1}{c_2}t_{l+1}(x),\quad a_2(x)=\frac{1}{c_2}r_{l+1}(x),$$
with $c_1$ and $c_2$ are the highest coefficients of $r_l(x)$ and
$s_{l+1}(x)$ respectively. In fact, Equation \eqref{eea1} is equivalent to 
$$
\begin{array}{r}
\overbrace{\phantom{.mmmmmmm}}^{n}\quad\overbrace{\phantom{.mmmmmm}}^{n-1}\quad\\
\begin{array}{r}
\left.
\begin{array}{l}
{}_{\displaystyle{n-1}}\\\phantom{r}
\end{array}\right\{
\\\phantom{r}\\
\left.\begin{array}{l}
\phantom{r}\\n\\\phantom{r}
\end{array}\right\{\\\phantom{r}\\
\left.\begin{array}{l}
{}_{\displaystyle{n-1}}\\\phantom{r}
\end{array}\right\{
\end{array}
\left(
\begin{array}{ccc|ccc}
t_{-n+1}&&&&&\\
\vdots&\ddots&&&&\\
\hline
t_0&\dots&t_{-n+1}&&&\\
\vdots&\ddots&\vdots&&&\\
t_{n-1}&\dots&t_0&&&\\
\hline
&\ddots&\vdots&\;\;1\;\;&&\\
&&&&\;\ddots\;&\\
&&t_{n-1}&&&\;\;1\;\;
\end{array}
\right)
\end{array}
\begin{pmatrix}
\phantom{r}\\u_1\\\phantom{r}\\b_1\\\phantom{r}
\end{pmatrix}
=\begin{pmatrix}\phantom{r}\\a_1\\\phantom{r}\\\hline 1\\0\\\vdots\\0\end{pmatrix}
$$
since $T$ is invertible then the $(2n-1)\times(2n-1)$ block at the bottom is invertible and then $u_1$ and $b_1$ are unique. Therefore $a_1$ is also unique. 

As $\deg r_l=n-1$ then, by Proposition \ref{eea},  $\deg s_{l+1}=(2n-1)-(n-1)=n$ and $\deg t_{l+1}=(2n-2)-(n-1)=n-1$. By the same proposition, we also have $\deg s_l\leq n-1$ and $\deg t_l\leq n-2$. 

Therefore, $\deg u_1=\deg s_l$ and $\deg b_1=\deg t_l$. Then
as $u_1 (x)$ and $\frac{1}{c_1} s_l$ are unitaries, 
$\frac{1}{c_1} s_l (x) =u_1 (x)$ which implies that $\frac{1}{c_1} t_l (x) =b_1 (x)$. For
the same reasons, we have $x^n+u_2(x)=\frac{1}{c_2}s_{l+1}(x)$ and
$b_2(x)=\frac{1}{c_2}t_{l+1}(x)$.

Finally, $Tu=e_1$ if and only if there exist $v(x)$, $w(x)$ such that  
\begin{equation}
\Tl(x)u(x)+x^nv(x)+(x^{2n}-1)w(x)=1.
\end{equation}
As $\Tl(x)=T^++x^{2n}T^-=T+(x^{2n}-1)T^-$, we deduce that
\begin{equation}\label{syz}
T(x)u(x)+x^nv(x)+(x^{2n}-1)(w(x)+T^-(x)u(x))=1.
\end{equation}
Moreover, we also have $T(x)u(x)-x^{-n+1}a_1(x)+x^nb_1(x)=1$ and
$x^{-n+1}a_1(x)=x^n(x\,a_1)-x^{-n}(x^{2n}-1)x\,a_1$. Thus \begin{equation}\label{syz2}T(x)u(x)+x^{n}(b(x)-x\, a(x))+(x^{2n}-1)x^{-n+1}a(x)=1.\end{equation}
Comparing \eqref{syz} and \eqref{syz2}, and as $1=x^nx^n-(x^{2n}-1)$
we deduce that $w(x)=x^{-n+1}a(x)-T_-(x)u(x)+1$, which is the part of
positive degree of $-T_-(x)u(x)+1$. This conclude the proof of the proposition.
\end{proof}

\begin{rem}
The usual Euclidean gcd algorithms are of computational complexity $\O(n^2)$, but superfast euclidean gcd algorithms use no more then $\Oc(n\,log^2 n)$ operations, exist. See for example \cite{MR2001757} chapter 11.
\end{rem}

The second method for computing $(u_1,v_1,w_1)$ and $(u_2,v_2,w_2)$ is
of polynomials and interpolation points.
We are interested in computing the
coefficients of the canonical basis element $\sigma_1,\,\sigma_2$ in
this basis. The coefficients of $\sigma_3$ can be
obtained by reduction of $(\Tl(x)\,x^n)\,B(x)$ by $x^{2n}-1$ where
$$ 
B(x)=
\begin{pmatrix}u(x)&u'(x)\\v(x)&v'(x)\end{pmatrix},
$$
where $(u,v),\,(u',v')$ are the two first coordinates of the 
solution of Equations \eqref{base1} and \eqref{base2}. 
A superfast algorithm for computing $B(x)$ is given in
\cite{MR1871324}. Let us describe how to compute it.

By evaluation of \eqref{base3} at the roots $\omega_j\in \Unit{2n}$,
we deduce that $(u(x), v(x))$ and $(u'(x), v'(x))$ are the solution of the following rational interpolation problem:
$$\left\{\begin{array}{l}\Tl(\omega_j)u(\omega_j)+\omega_j^nv(\omega_j)=0\\
\Tl(\omega_j)u'(\omega_j)+\omega_j^nv'(\omega_j)=0\end{array},\right. $$
with
$$\left\{\begin{array}{l}u_n=1,\,v_n=0,\\u'_n=0,\,v'_n=1.\end{array}\right.$$

\begin{defn}
The $\tau$-degree of a vector polynomial $w(x)=(w_1(x)\,w_2(x))^T$ is defined as 
$$\tau-\deg w(x):=\max\{\deg w_1(x),\,\deg w_2(x)-\tau\}$$
\end{defn}

\begin{defn}
A polynomial vector in $\kk[x]^{2}$ is called $\tau$-reduced if the $\tau$-highest degree coefficients are linearly independent. 
\end{defn}
By construction, the columns of $B(x)$ form a $n$-reduced basis of the
module of polynomial vectors $r(x)\in\kk[x]^2$ that satisfy the interpolation conditions
$$f_j\, r(\omega_j)=0,\;\;j=0,\ldots,2n-1$$ with $f_j=
(\Tl(\omega_j),\omega^n_j)\in \kk^{2}$.
The columns of $B(x)$ are also called a $n$-reduced basis for the interpolation data $(\omega_j,f_j),\,j=0,\ldots,2n-1$.

\begin{thm}
Let $\tau=n$ and $J$ be a positive integer. Let
$\lambda_1,\ldots,\lambda_J\in\kk$ and $\phi_1,\ldots,\phi_J\in\kk^2\setminus\{(0,0)\}$. 
Let $1\leq j\leq J$ and $\tau_J\in\mathbb{Z}$. Suppose that $B_j(x)\in\kk[x]^{2\times2}$ is a
$\tau_J$-reduced basis matrix with basis vectors having
$\tau_J-$degree $\delta_1$ and $\delta_2$, respectively, corresponding
to the interpolation data $\{(\lambda_i,\phi_i); i=1,\ldots,j\}$. 

Let $\tau_{j\rightarrow J}:=\delta_1-\delta_2$. Let $B_{j\rightarrow
  J}(x)$ be a $\tau_{j\rightarrow J}$-reduced basis matrix
corresponding to the interpolation data $\{(\lambda_i, \phi_i \,B_j(\lambda_j)); i=j+1,\ldots,J\}$.

Then $B_J(x):=B_j(x)B_{j\rightarrow J}(x)$ is a $\tau_J$-reduced basis matrix corresponding to the interpolation data  $\{(\lambda_i,\phi_i); i=1,\ldots,J\}$.
\end{thm}
\begin{proof}
For the proof, see \cite{MR1871324}.
\end{proof}

When we apply this theorem with $\lambda_{j}=\omega_j\in\Unit{2n}$ as
interpolation points, we obtain a superfast algorithm
in $\mathcal{O}(n\log^2n)$ to compute $B(x)$. See \cite{MR1871324} for
more details.

\section{Conclusion}
In this paper, we re-investigate the solution of a Toeplitz system $T\,
u =g$ from a new point of view, by correlating the solution of such a
problem with generators of the syzygy module  $\ML(\Tl(x), x^{n},
x^{2n}-1)$ associated to the Toeplitz matrix $T$. 
We show that $\ML(\Tl(x), x^{n}, x^{2n}-1)$ is free of rank $2$
and that it has a $n$-basis. We show that finding a $n$-basis of
$\ML(\Tl(x),x^{n},x^{2n}-1)$ is equivalent to computing the
Gohberg-Semencul decomposition of $T^{-1}$, and  we reduce the
solution of $T\,u=g$ to an Euclidean division.
We give two superfast algorithms computing a $n$-basis of
$\ML(\Tl(x),x^{n},x^{2n}-1)$ and a superfast algorithm to obtain the
solution from this $n$-basis.

A perspective of this work is to generalize the approach to two-level
Toeplitz systems or to Toeplitz-block-Toeplitz matrices and to
correlate the basis computation of a multivariate syzygy module to 
``Gohberg-Semencul'' decompositions for Toeplitz-block-Toeplitz matrices.


\end{document}